\documentclass[12]{article}
\usepackage{amsmath, amssymb, amsthm, color, hyperref}
\usepackage{mathrsfs}
\usepackage{graphicx}

\newtheorem{theorem}{Theorem}[section]
\newtheorem{definition}[theorem]{Definition}

\newtheorem{prop}[theorem]{Proposition}
\newtheorem{coro}[theorem]{Corollary}

\newcommand{\ba}{\mbox{\boldmath $a$}}
\newcommand{\sA}{{\cal A}}
\newcommand{\sS}{{\cal S}}

\def\mAth{\mathsurround=0pt}
\def\eqalign#1{\,\vcenter{\openup1\jot \mAth
\ialign{\strut\hfil$\displaystyle{##}$&$\displaystyle{{}##}$\hfil
    \crcr#1\crcr}}\,}

\newcommand{\eq}{\begin{equation}}
\newcommand{\eeq}{\end{equation}}
\renewcommand{\bar}{\overline}
\newcommand{\base}{4}
\begin{document}
\title{\bf 
On the Complexity of
Combinatorial Optimization on Fixed Structures}
\author{Nimrod Megiddo\thanks{IBM Almaden Research Center, San Jose, California, \url{http://stanford.io/2hPN3sG}}}
\date{November 16, 2024}
\maketitle
\begin{abstract}
Combinatorial optimization can be described as the problem of finding a feasible subset  that maximizes a objective function.
The paper discusses combinatorial optimization
problems, where for each dimension the set of feasible subsets is fixed. It is demonstrated that in some cases fixing the structure makes the problem easier, whereas in general the problem remains NP-complete.
\end{abstract}

\section{\hskip -16pt. Introduction}
An optimization problem of the form 
$\mbox{Maximize}~ \sum_{j=1}^n c_j x_j$ 
subject to 
$\sum_{j=1}^n a_{ij} x_j \le b_i$, $i=1,\ldots,m$,
and $x_j \in \{0,1\}$, $j=1,\ldots,n$,
can be viewed as maximizing $c(S)=\sum_{j\in S}x_j$
over a collection $\sS$ of subsets 
$S \subseteq N=\{1,\ldots,n\}$.  
Most of complexity analysis is asymptotic as $n$ tends to infinity.
We are concerned here with the question of the complexity
of combinatorial problems when the underlying structure is fixed. 
We note that fixing the structure means that for every $n\in \mathbb{N}$ the structure is unique. 

An example where the underlying structure may be fixed is the {\em shortest weight-constrained path} problem \cite{garey1979computers} 
(Problem [ND30], Megiddo, 1977).  The problem is posed over a directed graph $G=(V,E)$, where each edge $e$ has a length $a(e)$, which is fixed, and a traversal time $c(e)$, which may vary. The problem is to find a path of least traversal time that is not longer than a given bound $b$.  The knapsack problem can be reduced to constrained shortest-path, hence the latter is NP-complete.

It is worth noting that constraints
$\sum_{j=1}^n a_j x_j \le b$ provide more information than just the collection of feasible subsets. For example, if two subsets
$S$ and $T$ have 
$\sum_{j\in S}a_{ij} = \sum_{j\in T}a_{ij}$ for every $i$, then this can become useful for a dynamic programming algorithm, as in known in the case of the knapsack problem.

We consider, for example,
Knapsack Feasibility, which is the following inequalities problem in $0,1$-variables $x_1,\ldots,x_n$:
\begin{equation} \label{knap:0}
\eqalign{ 
\sum_{j=1}^n a_j x_j \le &\ b    \cr
\sum_{j=1}^n c_j x_j \ge &\ d~,  \cr}
\end{equation} 
where  the weights $a_1,\ldots,a_n$, 
the capacity $b$, 
the rewards $c_1,\ldots,c_n$, 
and the total-reward bound $d$,
are given integers.  
Obviously, the structure here is 
\[ \sS = \{ S\subseteq N~:~ \sum_{j\in S} a_j \le b \}~.\]
The problem was one of the first proven by Karp \cite{karp1972reductibility} to be NP-complete.
It is NP-complete even when $a_j=c_j$, $j=1,\ldots,n$, and
$b=d$, in which case it is called the Subset-Sum problem. A special case of the latter, where $b= \sum_{j=1}^n a_j /2$, is called the Partition problem, which is also NP-complete.  
Hard instances of the knapsack problem were identified in \cite{chvatal1980hard}.

The {\em Fixed-Weights Knapsack Feasibility} problem is defined as follows. 
\begin{definition}\rm
{\rm [FW-Knapsack]}
Suppose $b^n$ and the coefficients $a_1^n,\ldots,a_n^n$ are fixed for every value of the number of variables $n$, and let the input instance consist only of $c_1,\ldots,c_n$ and $d$.
Given $c_1,\ldots,c_n$ and $d$, recognize whether or not 
the system  {\rm (\ref{knap:0})}  has a $0,1$-solution.
\end{definition}
We will use both $b$ and $a_0$ to denote the capacity, i.e., $a_0^n = b^n$.

The question we are concerned with in this paper is whether fixing the weights makes the problem easier than the standard one.  Obviously, in special cases the problem is easy, for example, when $a_1=\cdots=a_n$. 
Also,
since the $a_i$s are fixed, they can be preprocessed and the time it takes to preprocess them is not included in the time complexity of the fixed-weights problem.

\section{\hskip -16pt. Improved time bounds for special classes}
Suppose a particular FW-knapsack problem is defined with respect to a sequences of tuples 
$\ba^n = (a^n_1,\ldots,a^n_n)$, $n\in \mathbb{N}$.

Let $n$ be fixed, and for convenience we sometimes omit the superscript $n$.
In this section, we consider the optimization version of the problem:
\begin{equation}
\eqalign{
    \mathop{\rm Maximize~} & \sum_{j=i}^n c_i x_i \cr
    \mbox{\rm subject to ~} & \sum_{i=1}^n a_i x_i  \le b \cr
    &\ x_i\in\{0,1\} ~~~(i=1,\ldots,n) ~. \cr
    }
\end{equation}
\subsection{Maximal feasible subsets}
Denote by $\sS^n$ the set of maximal subsets among all subsets $S\subseteq N$
that satisfy $\sum_{j\in N} a_j \le b$.
That is, subsets $S \subseteq N$ such that
$\sum_{j\in S} a^n_j \le b^n$ and for every $T \supset S$, $\sum_{j\in S} a^n_j > b^n$. 
Trivial time bounds can be obtained based on the number of maximal feasible subsets.
Specifically, denote by $f(n)$ the number of 
members of $\sS^n$. 
The problem can obviously be solved by checking all the member of $\sS^n$ and finding one  $S\in \sS^n$ that maximizes $\sum_{j\in S} c_i$.  
This can be done in $O(n\cdot f(n))$ time.

Note that the production of the  list of all maximal feasible subsets is carried out prior to receiving the input instance $(c_1,\ldots,c_n)$.  
The running time of this preprocessing step can be much larger than the number $f(n)$.
From this aspect, the FW-knapsack problem can sometimes be easier than the standard one.
\subsection{Dynamic-Programming bounds}
Less trivial bounds can be obtained from dynamic-programming algorithms.
For every $j$, $j=1,\ldots,n$, denote by 
\[ \sA(n,j) = \{A(1;\,j), \ldots, A(N(n,j);\,j)\} \]
the set of all the distinct sums of weights $a_{\ell}$ over subsets of
$\{1,\ldots,j\}$.
For every $j$, $j=1,\ldots,n$, and $ S \in \sA(n,j)$, denote by $\nu(S,j)$ the index that satisfies
\[  S  = A(\nu(S,j);\,j) ~.\]

Denote by $P(j,\,\ell) = P(j,\,\ell;\, n)$ the following residual sub-problem:
\begin{equation}
\eqalign{
    \mathop{\rm Maximize~} & \sum_{i=j+1}^n c_i x_i \cr
    \mbox{\rm subject to ~} & \sum_{i=j+1}^n a_i x_i  \le b - A(\ell;\,j)   \cr
   &\ x_i\in\{0,1\} ~~~(i=j+1,\ldots,n) ~ \cr
}
\end{equation}
and let $F(j,\,\ell) = F(j,\,\ell;\, n)$ denote the optimal objective-function value of $P(j,\,\ell)$.  

Consider the problem  $P(j,\,\ell)$. 
Distinguish two cases:\\
Case I. We choose $x_j=1$, so we get a reward of $c_j$ and the residual capacity is reduced from $b-A(\ell;\,j)$
to $b - A(\ell;\,j) - a_j$.  
There exists $\ell'$ such that 
\[ A(\ell;\,j) + a_j = A(\ell';\,j+1) ~.\]
In fact, according to the above-defined notation,
\[ \ell' 
= \nu(A(\ell;\,j) + a_j, ~ j+1)~.\]
Thus, in this case the problem is reduced to $P(j+1,\,\nu(A(\ell;\,j) + a_j)) $.\\
Case II.  We choose $x_j=0$.  
In this case, we do not get any reward and the residual capacity does not change, so the problem is reduced to 
$P(j+1,\,\nu(A(\ell;\,j))) $.\\
It follows that the dynamic-programming recursive formula is:
\[ F(j,\,\ell) 
= \max \left\{ 
c_j + F(j+1,\,\nu(A(\ell;\,j) + a_j)),~
F(j+1,\,\nu(A(\ell;\,j))) 
\right\} ~.
\]

\section{\hskip -16pt. NP-completeness} 
For the proof of NP-completeness, we will employ here a reduction from the following variant of 3-SAT:
\begin{definition}\rm
[1-in-3-SAT] \cite{wikip}
Given a CNF formula with at most three literals per clause, recognize whether it has a satisfying truth-value assignment
so that each clause has exactly one true literal.
\end{definition}
  
This problem is known to be NP-complete
\cite{schaefer1978complexity, wikip}.
In the Appendix, we include a proof of that fact based on the sketch in \cite{wikip}.

We reduce 1-in-3-SAT to the knapsack  feasibility problem (\ref{knap:0}), 
where the weights are fixed for every value of $n$.
The reduction is similar to the standard reduction of 3-SAT to the subset-sum problem, using large coefficients.  
Each of the inequalities in (\ref{knap:0}) can emulate a set of inequalities, 
each of which involving a different set of digital positions in the coefficients.

\subsection{An optimization model for 1-in-3-SAT}
Denote by $k$ and $m$ the number of variables and the number of clauses, 
respectively, in the input instance of 1-in-3-SAT. 

We first develop a special  $0,1$-optimization model for 1-in-3-SAT, which is amenable to constraint aggregation, so that the knapsack weights are independent of the particular CNF formula.
Suppose the given instance of 1-in-3-SAT involves the set of literals
$Z = \{z_1,\bar z_1,\ldots,z_k,\bar z_k\}$,
and has $m$ clauses.

\subsubsection{Variables}
For each literal in $Z$, we will have a corresponding decision variable.  
Denote these decision variables by 
$$x_1,\bar x_1,\ldots,x_k,\bar x_k~.$$

We will also employ additional $0,1$-variables 
$x_{ij}$ and $\bar x_{i,j}$, 
$i=1,\ldots,m$ and $j=1,\ldots,k$,
as follows.  
We will have $x_{ij}=1$
if and only if 
the variable $x_j$ appears in the clause $C_i$, and $x_j=1$. 
Similarly, $\bar x_{ij}=1$
if and only if 
the variable $\bar x_j$ appears in the clause $C_i$, and $\bar x_j = 1$.
Finally, additional slack variables 
$s_{ij}$ and $\bar s_{ij}$
are introduced below.

\subsubsection{Constraints}
The model has the following ``unique-choice'' constraints, which do not depend on the particular input CNF formula:
\eq \label{constr:eq}
\eqalign{
x_j + \bar x_j 
=&\ 1~~~~~(j=1,\ldots,k)~\cr
\sum_{j=1}^k x_{ij}
+ \sum_{j=1}^k \bar x_{ij}
=&\ 1  ~~~~~(i=1,\ldots,m)~,\cr
}
\eeq
and the following inequality constraints, which also do not depend on the particular CNF formula: 
\eq  \label{constr:ineq}
\eqalign{
x_j \ge 
&\ x_{ij}
  ~~~~~(i=1,\ldots,m,~j=1,\ldots,k)\cr
\bar x_j \ge 
&\ \bar x_{ij}
  ~~~~~(i=1,\ldots,m,~j=1,\ldots,k)~.\cr
}
\eeq
For every clause $C_i$, we also impose a constraint as follows.
Suppose 
$C_i = u_i \vee v_i \vee w_i$, 
where\footnote{Without ambiguity, we also write $C_i = \{u_i,v_i,w_i\}$.} 
$$\{u_i, v_i, w_i\} \subset Z = \{z_1,\bar z_1,\ldots,z_k,\bar z_k\}
~~~~(i=1,\ldots,m)~.$$
Denote for every literal $z\in Z$,
\[ \delta(i,z) = 
\begin{cases}
    1 &  \mbox{ if $z \in C_i$}  \\
    0 &  \mbox{ if $z \not\in C_i$} ~.
\end{cases}
\]
The constraints that ensure satisfaction of the 
individual clauses will be
\eq \label{const:Ci}
\sum_{j\,:\, x_j \in C_i} x_{ij}
+ \sum_{j\,:\,\bar x_j \in C_i}
  \bar x_{ij} \ge  1 ~~~~~(i=1,\ldots,m)~.
\eeq
Finally, we impose the constraints
\eq \label{constr:tij}
x_{ij} \ge x_j ~~~
                (i=1,\ldots,m,~
                 j=1,\ldots,k \,:\,
                 x_j \in C_i )
\eeq
and
\eq \label{constr:tijbar}
\bar x_{ij} \ge \bar x_j ~~~
                (i=1,\ldots,m,~
                 j=1,\ldots,k \,:\,
                 \bar x_j \in C_i )~.
\eeq
The constraints 
(\ref{constr:tij}) 
and 
(\ref{constr:tijbar})
will be converted, respectively, to
\eq \label{constr:tijconverted}
(x_{ij}  + \bar x_j)\, 
\delta(i, x_j) 
\ge 
\delta(i, x_j)  ~~~~~(i=1,\ldots,m,~
                 j=1,\ldots,k)~.
\eeq
\eq \label{constr:tijbarconverted}
(\bar x_{ij}  + x_j )\,
\delta(i, \bar x_j)
\ge 
\delta(i, \bar x_j) ~~~~~(i=1,\ldots,m,~
                 j=1,\ldots,k)~.
\eeq
Note that when $x_j \notin C_i$, then the constraint (\ref{constr:tijconverted})
is redundant, 
and when $\bar x_j \notin C_i$, then the constraint (\ref{constr:tijbarconverted})
is redundant.  
The reason why we still use these constraints is because we have to fix the digital positions corresponding to individual constraints independently of the specific CNF formula.

\begin{prop}   \label{model_valifity}
A CNF formula with at most three literals per clause has a satisfying assignment with exactly one literal per clause if and only if the constraints 
(\ref{constr:eq}), 
(\ref{constr:ineq},
(\ref{const:Ci}),
(\ref{constr:tij}) 
and (\ref{constr:tijbar}) 
are satisfied
\end{prop}
\begin{proof}
The `only if' part of obvious.
For the `if' part, suppose the constraints are satisfied.
The constraints (\ref{constr:eq}) and 
(\ref{const:Ci}) imply that the truth-value assignments defined by the $x_j$s and 
$\bar x_j$s satisfy all the clauses.
Also, the constraints (\ref{constr:tij})
imply that 
if $x_j=1$ and $x_j\in C_i$, then
$x_{ij}=1$.  
By the constraints
(\ref{constr:ineq}), the latter implies that 
for every $j'\not=j$, 
$x_{ij'}=\bar x_{ij'}=0$, 
and  also
$\bar x_{ij} = 0$. 
Thus, if $x_{j'} \in C_i$, 
then we have $x_{j'}=0$, 
and if $\bar x_{j'} \in C_i$, 
then we have $\bar x_{j'}=0$.

Similarly, by (\ref{constr:tijbar})
if $\bar x_j=1$ and $\bar x_j\in C_i$, 
then $\bar x_{ij}=1$, 
and hence for every $j'\not=j$, 
$x_{ij'}=\bar x_{ij'}=0$, and also
$x_{ij}= 0$. 
Thus, if $\bar x_{j'} \in C_i$, 
then we have $\bar x_{j'}=0$, 
and if $x_{j'} \in C_i$, 
then we have $x_{j'}=0$.

All of the above implies that in the 
truth-value assignment defined by the $x_j$s and $\bar x_j$s exactly one literal per clause is true.
\end{proof}

\subsubsection{Converting to unique-choice constraints}
The constraint 
\[ x_i \ge x_{ij}  \]
can be turned into a unique-choice constraint using a slack variable $s_{ij}$ as follows.
\eq \label{constr:sij}
x_{ij} + s_{ij} + \bar x_j = 1 ~.
\eeq
Similarly, the constraint
\[ \bar x_i \ge \bar x_{ij} \]
can be turned into a unique-choice constraint using a slack variable 
$\bar s_{ij}$ as follows.
\eq \label{constr:barsij}
\bar x_{ij} + \bar s_{ij} + x_j = 1 ~.
\eeq

\subsection{Aggregating constraints}
Reductions of integer-programming problems to the knapsack problem were proposed as early as 1971 \cite{bradley1971transformation}.

\renewcommand{\base}{\beta}
\begin{prop} \label{prop:aggreg}  
The following system of $m$ equality constraints  
with $0,1$-coefficients $a_{ij}$
and $n$ $0,1$-variables
$x_1,\ldots,x_n$ 
\begin{equation} \label{many}
\sum_{j=1}^n a_{ij} x_j =1 ~~~~~(i=1,\ldots,m)~,
\end{equation}
where there are at most $p$ nonzeros coefficients per equality,
is equivalent to the following single 
equality constraint:
\begin{equation} \label{sing} 
\sum_{i=1}^m \base ^{i-1} \bigg(\sum_{j=1}^n a_{ij} x_j\bigg) = \sum_{i=1}^m \base ^{i-1} 
=  \frac{\base ^m - 1 }{\base  - 1}~,
\end{equation}
where $\beta = p+1$.
\end{prop}
The proof is given in Appendix A.

\subsection{The reduction}
In the reduction to a fixed-weights knapsack problem,
for every pair $(k,m)$, we set the number of variables in the knapsack problem to 
$n = 2k + 4 k\,m$ ($2k$ variables of the type
$x_j$ or $\bar x_j$ and $k\,m$ variables of each of the types $x_{ij}$, $\bar x_{ij}$,
$s_{ij}$ and $\bar s_{ij}$).
We will fix an $n$-tuple of weights $(a^{(k,m)}_{1},\ldots,a^{(k,m)}_{n})$
and a capacity $b^{(k,m)}$, 
so that for every given instance
of 1-in-3-SAT with the same numbers of variables and clauses, the reduction uses the same weights 
$a_j = a^{(k,m)}_{j}$, $j=1,\ldots,n$, 
and capacity $b = b^{(k,m)}$.
To avoid a situation where the same number $k$ occurs in multiple different pairs $(k,m)$, 
we may assume, without loss of generality, that $k=m$.

We first apply Proposition \ref{prop:aggreg}
to the family of the constraints (\ref{constr:eq}), 
(\ref{constr:sij}) 
and (\ref{constr:barsij}).  
Thus, we obtain a linear function
\[
\eqalign{
A = A((x_j),(\bar x_j), &
(x_{ij}), (\bar x_{ij}),
(s_{ij}), (\bar s_{ij})) \cr
= &\ \sum_{j=1}^k \bigg( a(x_{j})x_{j} + a(\bar x_j) \bar x_{j}
+\sum_{i=1}^m \bigg(
  a(x_{ij})x_{ij} +  a(\bar x_{ij})\bar x_{ij} +
  a(s_{ij}) z_{ij}+  a(\bar s_{ij})\bar s_{ij}
  \bigg)\bigg)
}\]
and a scalar $b$ so that equation
\eq \label{const:sing_pre}
A((x_j),(\bar x_j),(x_{ij}),
(\bar x_{ij}),
(s_{ij}),
(\bar s_{ij}))
  = b
\eeq
is equivalent to the conjunction of those constraints.
We will represent (\ref{const:sing_pre})
as two inequalities:
\begin{eqnarray}
A((x_j),(\bar x_j),(x_{ij}),
(\bar x_{ij}),
(s_{ij}),
(\bar s_{ij}))
  &\le & b  \label{const_knp}
  \\
A((x_j),(\bar x_j),(x_{ij}),
(\bar x_{ij}),
(s_{ij}),
(\bar s_{ij}))
  &\ge & b ~. \label{const_rmv}
\end{eqnarray}
The first will serve as the knapsack weight constraint, whereas the second will be used for constructing the knapsack value constraint.

We aggregate the constraints 
(\ref{const:Ci}), 
(\ref{constr:tijconverted}
and (\ref{constr:tijbarconverted}) 
using the coefficients given in Table 1.
\begin{table}[h]
\begin{center}
\begin{tabular}{|c|c|c|}
\hline
& constraints & coefficient\\
\hline 
&&\\
 (\ref{const:Ci} )
 &
$\sum_{j\,:\, x_j \in C_i} x_{ij}
+ \sum_{j\,:\,\bar x_j \in C_i}
  \bar x_{ij} \ge  1$ 
&  $\beta^{i-1}$
\\
&&\\
\hline
&&\\
(\ref{constr:tijconverted})
& $(x_{ij}  + \bar x_j)\,\delta(i, x_j) 
\ge 
\delta(i, x_j)$
& $\beta^{m+k(i-1)+(j-1)}$
\\
&&\\
\hline
&&\\
(\ref{constr:tijbarconverted})
&
$(\bar x_{ij}  + x_j )\,\delta(i, \bar x_j)
\ge 
\delta(i, \bar x_j)$
&
$\beta^{m+km+k(i-1)+(j-1)}$
\\
&&\\
\hline
\end{tabular}
\end{center}
\caption{Aggregation coefficients}
\end{table}
Thus, denote
\eq
\eqalign{
C_0((x_j),(\bar x_j),&
(x_{ij}), (\bar x_{ij}),
(s_{ij}), (\bar s_{ij}))
=  \cr
&\ ~~~~ \sum_{i=1}^m \base^{i-1}
     \bigg(\sum_{j\,:\, x_j\in C_i}x_{ij}~
        + \sum_{j\,:\, \bar x_j\in C_i}x_{ij}
        \bigg) \cr
&\ + \sum_{i=1}^m\sum_{j=1}^k
   \base^{m + k\,(i-1) + (j-1)}
   ~(x_{ij} +\bar x_j)\,\delta(i, x_j) \cr
&\ + \sum_{i=1}^m\sum_{j=1}^k
   \base^{m + k\, m + k\,(i-1) + (j-1)}
   ~(\bar x_{ij} + x_j)) \,\delta(i, \bar x_j)
}
\eeq
and 
\eq
\eqalign{
d_0 ~=&\ ~ 
\sum_{i=1}^m \base^{i-1}
~+~ \sum_{i=1}^m\sum_{j=1}^k
   \base^{m + k\,(i-1) + (j-1)}
   \cdot \delta(i,x_j)
~+~ \sum_{i=1}^m\sum_{j=1}^k
   \base^{m + k\, m + k\,(i-1) + (j-1)}
   \cdot \delta(i,\bar x_j)\cr
\le &\ \sum_{i=0}^{ (2k+1)\, m  -1}\base^i
 ~~=~~ \frac{\base^{(2k+1)\,m} - 1}{\base -1}
}\eeq
and the aggregated constraint is
\eq  \label{const:C0d0}
C_0((x_j),(\bar x_j)
(x_{ij}), (\bar x_{ij}),
(s_{ij}), (\bar s_{ij})) \ge d_0 ~.
\eeq
Finally, we combine $A$ and $C_0$ into:
\eq \label{def:C}
C = \beta^{(2k+1)\,m} A + C_0 
\eeq
and $b$ and $d_0$ into
\eq \label{def:d}
d = \beta^{(2k+1)\,m}\, b + d_0 ~.
\eeq
Thus, we use the following:
\eq \label{finalknap}
C((x_j),(\bar x_j),
(x_{ij}), (\bar x_{ij}),
(s_{ij}), (\bar s_{ij})) \ge d
\eeq 
as the knapsack value constraint.

\begin{prop} \label{prop:3.4}
If the constraints 
{\em(\ref{const_knp})} and {\em(\ref{finalknap})}
are satisfied, then the constraints
{\em(\ref{const_rmv})} and {\em(\ref{const:C0d0})} are satisfied.
\end{prop}
\begin{proof}
Suppose the constraints 
(\ref{const_knp}) and (\ref{finalknap})
are satisfied.
Suppose the constraint (\ref{const_rmv}) is not satisfied.
Thus,
\[ A \le b - 1 \]
and therefore,
\[ \eqalign{
\base^{(2k+1)\, m } A + C_0
~~\le~~ &\ 
\base^{2k+1)\,m} (b-1) 
+  2k\,m \sum_{i=0}^{(2k+1)\,m - 1}\base^i\cr
~~<~~ &\ \base^{2k+1)\,m} (b-1) + \base^{(2k+1)\,m}
~~<~~ d ~.
}\]
It follows that
(\ref{finalknap}) is not satisfied, hence a contradiction. Thus, we proved that
(\ref{const_rmv}) is satisfied and we have established (\ref{const:sing_pre}), $A=b$.

Next,
\[ \eqalign{
d_0 
= &\  d - \base^{(2k+1)\,m}\, b  
~~~~~~~~\mbox{(by (\ref{def:d})\,)}\cr
\le 
&\ C - \base^{(2k+1)\,m}\, b  
~~~~~~~~\mbox{(by (\ref{finalknap}))}
\cr
= &\ \base^{(2k+1)\,m} A + C_0 
- \base^{(2k+1)\,m} b = C_0 ~.
~~~~~~~~\mbox{(by (\ref{def:C})\,)}
\cr
}\]
This proves that also (\ref{const:C0d0})
is satisfied.
\end{proof}
\begin{prop}
If the constraints {\em(\ref{const_rmv})} and 
{\em(\ref{finalknap})} are satisfied, then 
all the constraints described  in {\em(\ref{constr:eq}), (\ref{constr:sij})} and
{\em(\ref{constr:barsij})} are satisfied.
\end{prop}
\begin{proof}
This is a direct application of Proposition
\ref{prop:aggreg}.
\end{proof}
\begin{prop}
If the constraints {\em(\ref{const_knp})} and {\em(\ref{finalknap})}
are satisfied, then all of the constraints
enumerated in 
{\em (\ref{const:Ci}),
(\ref{constr:tijconverted})} and 
{\em(\ref{constr:tijbarconverted})}
are satisfied.
\end{prop}
\begin{proof}
Under the conditions of the present proposition, by Proposition \ref{prop:3.4}, also (\ref{const:C0d0}) is satisfied.

Note that by (\ref{constr:eq}),
\[ 
\sum_{j\,:\,x_j\in C_i} x_{ij}
+ \sum_{j\,:\,\bar x_j\in C_i} \bar x_{ij} 
\le \sum_{j=1}^k x_{ij}
+ \sum_{j=1}^k \bar x_{ij} 
= 1 ~~~~~(i=1,\ldots,m)
\]
and by (\ref{constr:eq}), (\ref{constr:sij}) and  (\ref{constr:barsij})
\[ x_{ij} + \bar x_{j} 
\le  x_j +\bar x_j = 1 
~~~~(i=1,\ldots,m,~j=1,\ldots,k)\]
and 
\[ \bar x_{ij} + x_{j} 
\le  \bar x_j + x_j = 1 
~~~~(i=1,\ldots,m,~j=1,\ldots,k)~.\]
The essence of the proof is as follows.
It follows from the definitions of $C_0$ and $d_0$ that there exist $U_i, V_i \in\{0,1\}$
such that $V_i \le U_i$ and 
\[ d_0 = \sum_{i=0}^{(2k+1)\, m-1} \base^i
 U_i \]
and 
\[
C_0 = \sum_{i=0}^{(2k+1)\, m-1} \base^i 
\cdot V_i ~.
\]
It is easy to see that
$C_0 \ge d_0$ only if $V_i = U_i$ for every $i$.
\end{proof}
Thus, we established the proposition below.
\begin{prop}
The weight constraint {\em(\ref{const_knp})} and
the value constraint {\em(\ref{finalknap})} have a feasible solution if and only if the given CNF formula has a satisfying assignment with exactly one true literal per clause.
\end{prop}

The digital positions are allocated to the various constraint as indicated in Table 2.

\begin{table}[h]
\begin{center}
\begin{tabular}{|c|c|c|c|}
\hline
& constraints & start position & number of positions \\
\hline 
&&& \\
 (\ref{const:Ci} )
 &
$\sum_{j\,:\, x_j \in C_i} x_{ij}
+ \sum_{j\,:\,\bar x_j \in C_i}
  \bar x_{ij} \ge  1$ 
&  $0$
& $m$ \\
&&& \\
\hline
&&& \\
(\ref{constr:tijconverted})
& $x_{ij}  + \bar x_j \ge 
\delta(i, x_j)$
& $m$
& $k\,m$
\\
&&&\\
\hline
&&& \\
(\ref{constr:tijbarconverted})
&
$\bar x_{ij}  + x_j \ge 
\delta(i, \bar x_j)$
&
$m+k\,m = (k+1)\,m$
&
$k\,m$
\\
&&&\\
\hline
&&&\\
(\ref{constr:sij})
&
$x_{ij} + s_{ij} + \bar x_j = 1$
&
$(k+1)\,m + k\,m = (2k+1)\,m$ 
&  $k\,m$  \\
&&&\\
\hline
&&&\\
(\ref{constr:barsij})
&
$\bar x_{ij} + s_{ij} + x_j = 1$
&
$(2k+1)\,m + k\,m = (3k+1)\,m$
&  $k\,m$ \\
&&& \\
\hline
&&& \\
(\ref{constr:eq})
&
$\sum_{j=1}^k x_{ij}
+ \sum_{j=1}^k \bar x_{ij}
= 1  $
&
$(3k+1)\, + k\,m = (4k+1)\, m$
&
$m$
\\
&&&\\
\hline
&&&\\
(\ref{constr:eq})
& $x_j + \bar x_j = 1$
&  $(4k+1)\, m + m = (4k+2)\,m $
&  $k$ \\
&&& \\
\hline
\end{tabular}
\caption{Allocating the digital positions}
\end{center}
\end{table}

Thus, we established to following:
\begin{theorem}
There exists a computable sequence of tuples 
$$
\{\ba^n\}_{n=1}^\infty 
= \{(a^n_0,a^n_1,\ldots,a^n_n)\}_{n=1}^\infty
$$
such that the FW-Knapsack problem with the constraint defined by $\ba^n$
in \em{(\ref{knap:0})}
is {\rm NP}-complete.
\end{theorem}

\subsection{An example}
Consider the instance $C_1\wedge C_2$ of {\rm 1-in-3-SAT} where 
\[ \eqalign {
C_1=&\ x_1 \vee x_2 \vee x_3 \cr
C_2 =&\ x_1 \vee \bar x_2  \vee \bar x_3~.\cr
} \]

Denote the constraints as follows.
\[\eqalign{
  A_j ~:~&\ x_j + \bar x_j = 1 ~~~~ (j=1,2,3) \cr
  C^0_i~:~& \ \sum_{j=1}^k x_{ij} + 
              \sum_{j=1}^k \bar x_{ij} = 1
              ~~~~(i=1,2)\cr
  \bar B_{ij}~:~&\ \bar x_{ij}+\bar s_{ij}+ x_j=1 
  ~~~~(i=1,2,~j=1,2,3) \cr
   B_{ij}~:~&\ x_{ij}+s_{ij}+\bar x_j=1 
  ~~~~(i=1,2,~j=1,2,3) \cr
  \bar D_{ij}~:~&\ (\bar x_{ij} + x_j)\delta(i,x_j) \ge \delta(i,x_j) 
   ~~~~(i=1,2,~j=1,2,3)  \cr
  D_{ij}~:~&\ (x_{ij} + \bar x_j)\delta(i,x_j) \ge \delta(i,x_j)
   ~~~~(i=1,2,~j=1,2,3) \cr
  }\] 
The weight coefficients are indicated by the table in Figure 1, where the first row in the table indicates the constraint and the second row indicates the power of $\base$ to which it corresponds.
\renewcommand{\ss}{}
\begin{figure} 
\label{figure:weight}
\includegraphics[width=420pt]{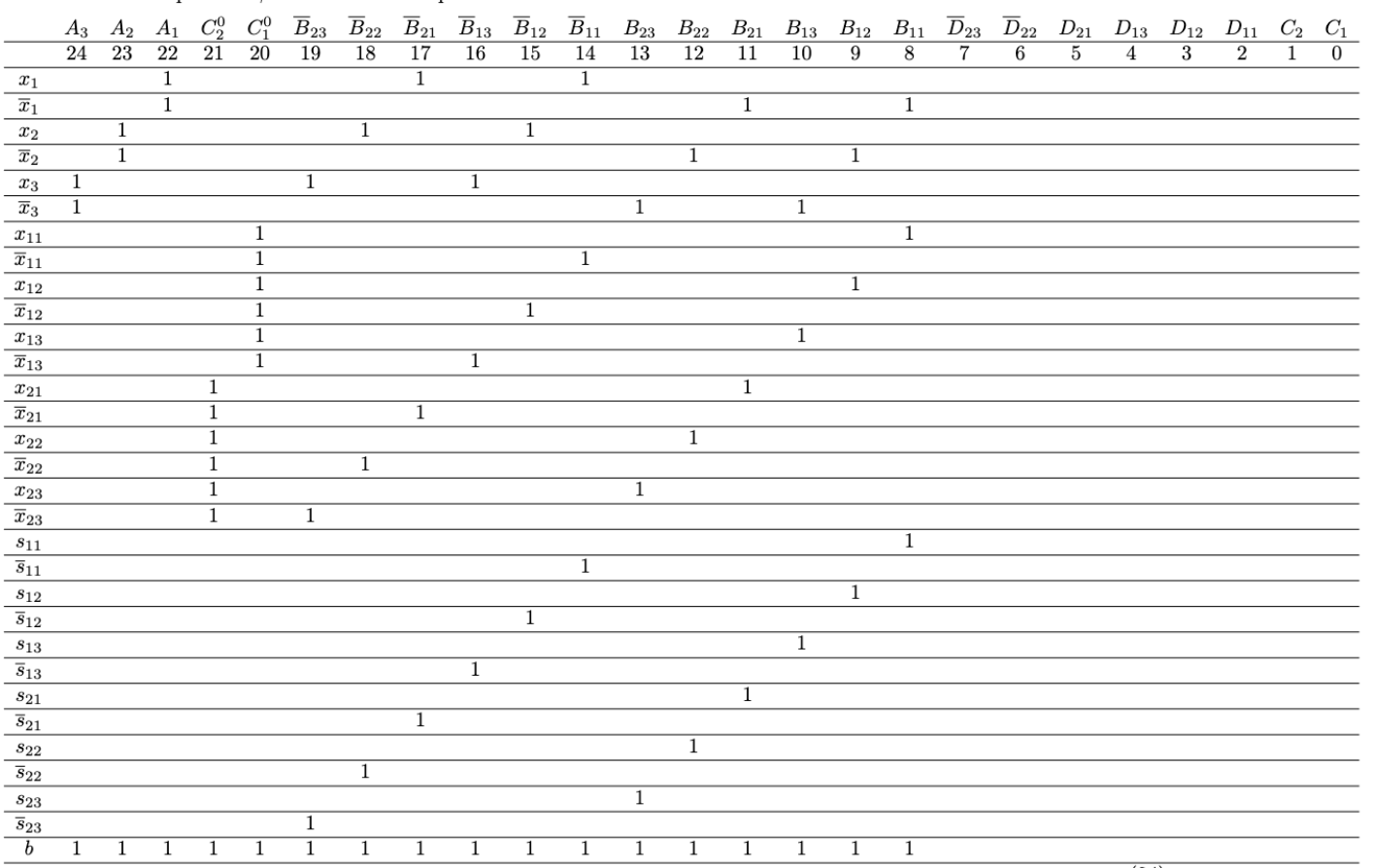}
\caption{Digits of the weight coefficients}
\end{figure}

The  value coefficients are indicated by the table in Figure 2.
\begin{figure} 
\includegraphics[width=420pt]{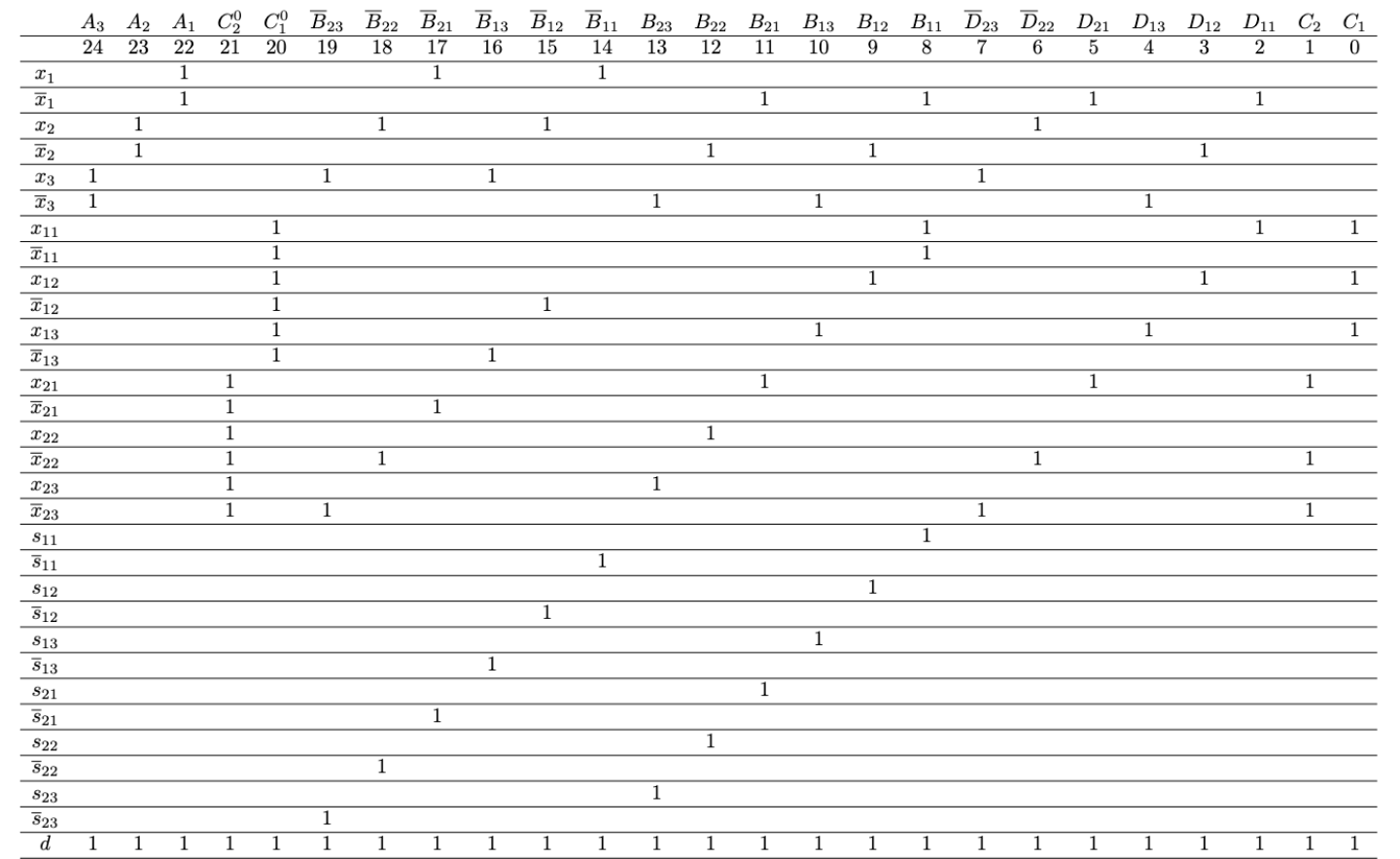}
\caption{Digits of the value coefficient}
\end{figure}

\vskip 6pt
\renewcommand{\base}{\beta}

\section{\hskip -16pt. Other problems}
The Covering Feasibility problem is essentially the same as the Knapsack Feasibility problem where the roles of the objective and constraint interchange. Our NP-completeness proof can be easily modified to prove that covering problem with fixed constraint is also NP-complete. 
We already mentioned the constrained shortest-path problem.

 \bibliographystyle{plain}

\section*{Appendix A}

{\it Proof of Proposition \ref{prop:aggreg}. }
Obviously, if $(x_1,\ldots,x_n)$ satisfies (\ref{sing}), then  it satisfies (\ref{sing}).

Conversely, suppose $(x_1,\ldots,x_n)$ satisfies (\ref{sing}). 
We first prove (\ref{many}) for $i=m$. 
It follows from (\ref{sing}) that
\[
\base^{m-1}\sum_{j=1}^n a_{mj} x_j 
= \frac{\base ^m -1}{\base  -1} -
\sum_{i=1}^{m-1} \base ^{i-1} \bigg(\sum_{j=1}^n a_{ij} x_j\bigg)  
<  \frac{\base ^m}{\base  -1} ~,
\]
hence,
\[
\sum_{j=1}^n a_{mj} x_j  
= \left\lfloor 
    \frac{\base  }
         {\base-1} 
    \right \rfloor   
= 1 ~.
\]
Also,
\[ \eqalign{
\base^{m-1}\sum_{j=1}^n a_{mj} x_j 
\ge &\ 
\frac{\base ^m - 1}{\base -1} 
- p \sum_{i=1}^{m-1} \base ^{i-1} \cr
=&\ \frac{\base ^m -1 }{\base  - 1} 
- p\cdot\frac{\base ^{m-1} - 1}{\base  -1} \cr
=&\ \frac{\base ^m - p\cdot\base ^{m-1} }{\base  - 1} 
 + \frac{ p - 1}{\base-1} \cr
\ge &\  \frac{\base^m - p\cdot\base^{m-1} }{\base  - 1} ~.
\cr
}\]
Hence, assuming $\base > p$,
\[
\sum_{j=1}^n a_{mj} x_j 
\ge  \left \lceil
\frac{\base  - p }{\base  -1}
\right\rceil= 1~,
\]
so we established that 
\[
\sum_{j=1}^n a_{mj} x_j   = 1~.
\]
We proceed by downward induction.  
For the inductive step, suppose $\ell \le m-1$ and for $i=\ell+1,\ldots,m$ we have 
\[ \sum_{j=1}^n a_{ij} x_j = 1~. \]
We have
\[\eqalign{
\base ^{\ell-1}\sum_{j=1}^n a_{\ell j} x_j 
=&\ \frac{\base ^m -1}{\base  -1} 
-\sum_{i=1}^{\ell-1} \base ^{i-1} \bigg(\sum_{j=1}^n a_{ij} x_j\bigg)
-\sum_{i=\ell+1}^{m} \base ^{i-1} \bigg(\sum_{j=1}^n a_{i j} x_j\bigg)  \cr
\le &\
\frac{\base ^m -1}{\base  -1} 
-\sum_{i=\ell+1}^{m} \base ^{i-1} \bigg(\sum_{j=1}^n a_{i j} x_j\bigg)  \cr
= &\
\frac{\base ^m -1}{\base  -1} 
-\sum_{i=\ell+1}^{m} \base ^{i-1}   \cr
= &\
\frac{\base ^m -1}{\base  -1} 
- \frac{\base ^{m} - \base ^{\ell}}{\base -1}
  \cr
\le &\
\frac{\base ^\ell}{\base  -1} 
  \cr
}\]
Hence,
\[
\sum_{j=1}^n a_{\ell j} x_j
\le  
\left\lfloor
\frac{\base ^\ell }{ \base ^{\ell- 1}\base   -1)} 
\right\rfloor 
= \left\lfloor
\frac{\base  }{ \base  -1} 
\right\rfloor = 1 ~.
\]
Also,
\[\eqalign{
\base ^{\ell-1}\sum_{j=1}^n a_{\ell j} x_j 
=&\ 
\frac{\base ^m -1}{\base  -1} 
-\sum_{i=1}^{\ell-1} \base ^{i-1} \bigg(\sum_{j=1}^n a_{ij} x_j\bigg)
-\sum_{i=\ell+1}^{m} \base ^{i-1} \bigg(\sum_{j=1}^n a_{i j} x_j\bigg)  \cr
\ge &\
\frac{\base ^m -1}{\base  -1} 
- p\sum_{i=1}^{\ell-1} \base ^{i-1} 
-\sum_{i=\ell+1}^{m} \base ^{i-1}   \cr
=&\
\frac{\base ^m -1}{\base  -1} 
- p \cdot \frac{\base ^{\ell-1}-1}{\base -1}
- \frac{\base ^{m} - \base ^{\ell}}{\base -1}
   \cr
=&\
\frac{p-1 
- p \,\base ^{\ell-1}
+ \base ^{\ell}}
{\base -1}
   \cr
>&\
   \frac{\base ^{\ell} - p \,\base ^{\ell-1}}
    {\base -1}\cr
\ge&\
\frac{\base ^{\ell} - (\base -1)\,\base ^{\ell-1}}
{\base -1}\cr
= &\
\frac{\base ^{\ell-1}}
{\base -1}\cr
}\]
It follows that
\[ \sum_{j=1}^n a_{\ell j} x_j
\ge \left\lceil
\frac{1}{\base -1}
\right\rceil = 1 ~,
\]
hence,
\[ \sum_{j=1}^n a_{\ell j} x_j  = 1 ~.
~~~~~~~~~~~~~~~~\qed
\]

\section*{Appendix B}
\begin{prop}
Given any $\{x,y,z\} \subseteq\{0,1\}$,
the system
\begin{equation} \label{eq:101}
\eqalign{
(1-x) + a + b = &\ 1 \cr
   y  + b + c = &\ 1 \cr
(1-z) + c + d = &\ 1 \cr
}
\end{equation}
has a solution 
$\{a,b,c,d\} \subseteq\{0,1\}$ 
if and only if
\[ x + y + z \ge 1 ~.\]
\end{prop}
\begin{proof}
First, if $(x,y,z)=(0,0,0)$, the system (\ref{eq:101}) reduces to the system
\begin{equation}
\eqalign{
a + b = &\ 0 \cr
b + c = &\ 1 \cr
c + d = &\ 0 \cr
}
\end{equation}
which has no $0,1$-solution. 
The other seven cases are as follows.
\begin{enumerate}
\item 
$(x,y,z)=(1,0,0)$. 
The system reduces to
\begin{equation}
\eqalign{
a + b = &\ 1 \cr
b + c = &\ 1 \cr
c + d = &\ 0 \cr
}
\end{equation}
and the solution is 
$(a,b,c,d)=(0,1,0,0)$.
\item 
$(x,y,z)=(0,1,0)$. 
The system reduces to
\begin{equation} 
\eqalign{
a + b = &\ 0 \cr
b + c = &\ 0 \cr
c + d = &\ 0 \cr
}
\end{equation}
and the solution is 
$(a,b,c,d)=(0,0,0,0)$.
\item 
$(x,y,z)=(0,0,1)$.
\begin{equation} 
\eqalign{
 a + b = &\ 0 \cr
 b + c = &\ 1 \cr
 c + d = &\ 1 \cr
}
\end{equation}
and the solution is 
$(a,b,c,d)=(0,0,1,0)$.
\item
$(x,y,z)=(1,1,0)$.
The system reduces to
\begin{equation} 
\eqalign{
 a + b = &\ 1 \cr
 b + c = &\ 0 \cr
 c + d = &\ 0 \cr
}
\end{equation}
and the solution is 
$(a,b,c,d)=(1,0,0,0)$.
\item 
$(x,y,z)=(1,0,1)$.
The system reduces to
\begin{equation} 
\eqalign{
 a + b = &\ 1 \cr
 b + c = &\ 1 \cr
 c + d = &\ 1 \cr
}
\end{equation}
and the solutions are 
$(a,b,c,d)=(1,0,1,0)$ and
$(a,b,c,d)=(0,1,0,1)$.
\item 
$(x,y,z)=(0,1,1)$.
The system reduces to
\begin{equation} 
\eqalign{
 a + b = &\ 0 \cr
 b + c = &\ 0 \cr
 c + d = &\ 1 \cr
}
\end{equation}
and the solution is
$(a,b,c,d)=(0,0,0,1)$.
\item 
$(x,y,z)=(1,1,1)$.
The system reduces to
\begin{equation} 
\eqalign{
 a + b = &\ 1 \cr
 b + c = &\ 0 \cr
 c + d = &\ 1 \cr
}
\end{equation}
and the solution is
$(a,b,c,d)=(1,0,0,1)$.
\end{enumerate}
\end{proof}

\begin{coro}
A disjunction $x\vee y \vee z$ is true
if and only if the following conjunction is satisfiable with exactly one true literal per clause:
\[ (\neg x \vee a \vee b)
\wedge (y \vee b \vee c)
\wedge (\neg z \vee c \vee d)~.\]
\end{coro}
\begin{coro}
The 3-SAT problem is polynomial-time reducible to 3-SAT with the requirement that in the satisfying assignment exactly one literal per clause is true;
hence the latter is NP-complete.
\end{coro}

\end{document}